\documentclass[11pt]{article}

\usepackage{fullpage}
\usepackage{latexsym}
\usepackage{amsfonts,amssymb}
\usepackage{amstext}
\usepackage{graphicx}
\usepackage{color}
\usepackage{url,hyperref}
\usepackage{amsmath,xspace}
\usepackage{shadow,shadethm}

 \newtheorem{lemma}{Lemma}[section]
 \newtheorem{theorem}[lemma]{Theorem}

\newenvironment{proof}{\vspace{-0.15in}\noindent{\bf Proof:}}%
        {\hspace*{\fill}$\Box$\par}

\newenvironment{proofof}[1]{\smallskip\noindent{\bf Proof of #1:}}%
        {\hspace*{\fill}$\Box$\par}

\newcommand{\opt}{\textrm{\sc OPT}\xspace}

\newcommand{\R}{{\mathbb R}}
\newcommand{\tsty}{}

\newcommand{\alg}{\ensuremath{{\sf Alg}}\xspace}

\newcommand{\bcp}{{\sf BCP}\xspace}
\newcommand{\alw}{{\sf ALW}\xspace}
\newcommand{\const}{\frac{2}{\epsilon}}
\newcommand{\constun}{\frac{4}{\epsilon}}
\newcommand{\dphidt}{{\frac{d \Phi_i(t)}{dt}}}
\newcommand{\dphidtside}{d \Phi_i(t)/dt}

\newcommand{\Xomit}[1]{ }

\newcommand{\initOneLiners}{%
    \setlength{\itemsep}{0pt}
    \setlength{\parsep }{0pt}
    \setlength{\topsep }{0pt}
}
\newenvironment{OneLiners}[1][\ensuremath{\bullet}]
    {\begin{list}
        {#1}
        {\initOneLiners}}
    {\end{list}}

\begin{document}

\title{Scalably Scheduling Power-Heterogeneous Processors\footnote{A preliminary version of this paper appeared in ICALP 2010}}

\author{
 Anupam Gupta\thanks{ Computer Science Department, Carnegie Mellon
    University, Pittsburgh, PA 15213, USA. 
Supported in part by
     NSF award CCF-0729022 and an Alfred P.~Sloan
     Fellowship.}
 \and Ravishankar Krishnaswamy$^\dagger$
%\fnmsp \footnotemark[1]
\and Kirk Pruhs\thanks{Computer Science Department,
University of Pittsburgh, Pittsburgh, PA 15260, USA. 
Supported in part by
 NSF grants CNS-0325353, IIS-0534531, and CCF-0830558, and
 an IBM Faculty Award.}
 }

\date{}
\maketitle
%\thispagestyle{empty}
%\vspace{-7mm}

\begin{abstract}
  We show that a natural online algorithm for scheduling jobs on a
  heterogeneous multiprocessor, with arbitrary power functions, is
  scalable for the objective function of weighted flow plus energy.
\end{abstract}

%\vspace{-0.4in}
\section{Introduction}
\label{sec:introduction}
%\vspace{-0.05in}

Many prominent computer architects believe that architectures consisting
of heterogeneous processors/cores, such as the STI Cell processor, will
be the dominant architectural design in the
future~\cite{Bower2008,Kumar2004,Kumar2006,Merritt2008,Tomer2006}.  The
main advantage of a heterogeneous architecture, relative to an
architecture of identical processors, is that it allows for the
inclusion of processors whose design is specialized for particular types
of jobs, and for jobs to be assigned to a processor best suited for
that job.  Most notably, it is envisioned that these heterogeneous
architectures will consist of a small number of high-power
high-performance processors for critical jobs, and a larger number of
lower-power lower-performance processors for less critical jobs.
Naturally, the lower-power processors would be more energy efficient in
terms of the computation performed per unit of energy expended, and
would generate less heat per unit of computation.  For a given area and
power budget, heterogeneous designs can give significantly better
performance for standard workloads~\cite{Bower2008,Merritt2008};
Emulations in \cite{Kumar2006} suggest a figure of 40\% better
performance, and emulations in \cite{Tomer2006} suggest a figure of
67\% better performance.
Moreover, even processors that were designed to be homogeneous, are
increasingly likely to be heterogeneous at run time~\cite{Bower2008}:
the dominant underlying cause is the increasing variability in the
fabrication process as the feature size is scaled down (although run
time faults will also play a role).  Since manufacturing yields would be
unacceptably low if every processor/core was required to be perfect, and
since there would be significant performance loss from derating the
entire chip to the functioning of the least functional processor (which
is what would be required in order to attain processor homogeneity),
some processor heterogeneity seems inevitable in chips with many processors/cores.

The position paper~\cite{Bower2008} identifies three fundamental
challenges in scheduling heterogeneous multiprocessors: (1)~the OS must
discover the status of each processor, (2)~the OS must discover the
resource demand of each job, and (3)~given this information about
processors and jobs, the OS must match jobs to processors as well as
possible. In this paper, we address this third fundamental challenge.
In particular, we assume that different jobs are of differing
importance, and we study how to assign these jobs to processors of
varying power and varying energy efficiency, so as to achieve the best
possible trade-off between energy and performance.

Formally, we assume that a collection of jobs arrive in an online
fashion over time. When a job $j$ arrives in the system, the system is
able to discover a \emph{size} $p_j \in \R_{> 0}$, as well as a
\emph{importance/weight} $w_j \in \R_{> 0}$, for that job. The
importance $w_j$ specifies an upper bound on the amount of energy that
the system is allowed to invest in running $j$ to reduce $j$'s flow by
one unit of time (assuming that this energy investment in $j$ doesn't
decrease the flow of other jobs)---hence jobs with high weight are more
important, since higher investments of energy are permissible to justify
a fixed reduction in flow. Furthermore, we assume that the system knows the
allowable speeds for each processor, and the system also knows the
power used when each processor is run at its
set of allowable speeds.  We make no real restrictions on the allowable speeds,
or on the power used for these speeds.\footnote{So the processors may or
  may not be speed scalable, the speeds may be continuous or discrete or
  a mixture, the static power may or may not be negligible, the dynamic
  power may or may not satisfy the cube root rule, etc.}  The online
scheduler has three component policies:
\begin{description}
\item[Job Selection:] Determines which job to run on each processor at any time.
\item[Speed Scaling:]Determines the speed of each processor at each time.
\item[Assignment:] When a new job arrives, it determines the processor
to which this new job is assigned.
\end{description}
%\kirknote{I changed ``dispatch'' to ``assignment'' everyplace to be
%consistent}

The objective we consider is that of \emph{weighted flow plus energy}.
The rationale for this objective function is that the optimal schedule
under this objective gives the best possible weighted flow for the
energy invested, and increasing the energy investment will not lead to a
corresponding reduction in weighted flow (intuitively, it is not
possible to speed up a collection of jobs with an investment of energy
proportional to these jobs' importance).

We consider the following natural online algorithm that essentially adopts the
job selection and speed scaling algorithms from the uniprocessor
algorithm in \cite{BCP}, and then greedily assigns the jobs based on
these policies.
\begin{shadebox}
\begin{description}
\item[Job Selection:] Highest Density First (HDF)
\item[Speed Scaling:] The speed is set so that the power is the fractional
  weight of the unfinished jobs.
\item[Assignment:] A new job is assigned to the processor
  that results in the least increase in the projected future weighted
  flow, assuming the adopted speed scaling and job selection policies, and
  ignoring the possibility of jobs arriving in the future.
\end{description}
\end{shadebox}
Our main result is then:
\begin{theorem}
  \label{thm:main1}
  This online algorithm is scalable for scheduling jobs on a
  heterogeneous multiprocessor with arbitrary power functions to
  minimize the objective function of weighted flow plus energy.
\end{theorem}
In this context, \emph{scalable} means that if the adversary can run
processor $i$ at speed $s$ and power $P(s)$, the online algorithm is
allowed to run the processor at speed $(1+\epsilon)s$ and power $P(s)$,
and then for all inputs, the online cost is bounded by $O(f(\epsilon))$
times the optimal cost.  Intuitively,
a scalable algorithm can handle almost the same load as
optimal;  for further elaboration, see~\cite{PST,Pruhs07}.  Theorem
\ref{thm:main1} extends theorems showing similar results for weighted
flow plus energy on a uniprocessor~\cite{BCP,Lachlan2009}, and for
weighted flow on a multiprocessor without power
considerations~\cite{Chadha2009}.  As scheduling on identical processors
with the objective of total flow, and scheduling on a uniprocessor with
the objective of weighted flow, are special cases of our problem, constant
competitiveness is not possible without some sort of resource
augmentation~\cite{LR,BC09}.

Our analysis is an amortized local-competitiveness argument. As is
usually the case with such arguments, the main technical hurdle is to
discover the ``right'' potential function. % A first strawman potential
% function to try is the potential function from~\cite{Chadha2009}, since
% the high-level intuition behind their algorithm is the same as in ours
% (greedily assign the job to the machine that does least harm, and run a
% good scheduling policy on each machine). Their global potential function
% is a sum of per-processor potential functions; though it does not take
% power/energy into consideration, one can define variants that work for
% on/off power curves .
The most natural straw-man potential function to try is the sum over all
processors of the single processor potential function used
in~\cite{BCP}.  While one can prove constant competitiveness with this
potential in some special cases (e.g. where for each processor the
allowable speeds are the non-negative reals, and the power satisfies the
cube-root rule), one can not prove constant competitiveness for general
power functions with this potential function.  The reason for this is
that the uniprocessor potential function from~\cite{BCP} is not
sufficiently accurate.  Specifically, one can construct
configurations where the adversary has finished all jobs, and where the
potential is much higher than the remaining online cost. This did not
mess up the analysis in \cite{BCP} because to finish all these jobs by
this time the adversary would have had to run very fast in the past,
wasting a lot of energy, which could then be used to pay for this
unnecessarily high potential.  But since we consider multiple
processors, the adversary may have no jobs left on a particular
processor simply because it assigned these jobs to a different
processor, and there may not be a corresponding unnecessarily high
adversarial cost that can be used to pay for this unnecessarily high
potential. % Another strawman solution is to try and extend the potential
% function of Chadha et al.~\cite{Chadha2009} to power-aware settings;
% we did not have success it can be extended to handle simple
% step-function-type power functions.

Thus, the main technical contribution in this paper is a seemingly
more accurate potential function expressing the additional cost required
to finish one collection of jobs compared to another collection of
jobs. Our potential function is arguably more transparent than the one used in~\cite{BCP},
and we expect that
this potential function will find future application in
the analysis of other power management algorithms.

In section~\ref{dsec:unweighted-flow}, we show that a similar online algorithm is
$O(1/\epsilon)$-competitive with $(1+\epsilon)$-speedup for
\emph{unweighted} flow plus energy.
We also remark that when the power functions $P_i(s)$ are
restricted to be of the form $s^{\alpha_i}$, our algorithms give a
$O(\alpha^2)$-competitive algorithm (with no resource augmentation needed) for the problem of minimizing weighted flow plus energy, and an
$O(\alpha)$-competitive algorithm for minimizing the unweighted flow plus energy,
where $\alpha = \max_i \alpha_i$.
%These extend the results of Lam et al. \cite{Lam08} who show scalable algorithms for the unweighted problem.

% \kirknote{I commented out the last sentence about comparing to Lam08. It seems the comparison between
% this result and the Lam08 results is not so straight-forward. I don't want to take a lot
% of space to go into this, and I certainly don't want to get it wrong.}

%\vspace{-0.15in}
\subsection{Related Results}
\label{sec:related-results}

Let us first consider previous work for the case of a single processor,
with unbounded speed, and a polynomially bounded power function $P(s) =
s^\alpha$.  \cite{PUW} gave an efficient offline algorithm to find the
schedule that minimizes average flow subject to a constraint on the
amount of energy used, in the case that jobs have unit work. \cite{AF}
introduced the objective of flow plus energy and gave a constant
competitive algorithm for this objective in the case of unit
work jobs. \cite{BPS} gave a constant competitive algorithm for the
objective of weighted flow plus energy.  The competitive ratio was
improved by \cite{LLTW08} for the unweighted case using a potential
function specifically tailored to integer flow. \cite{BCLL08} extended
the results of \cite{BPS} to the bounded speed model,
and~\cite{STACS2009} gave a \emph{nonclairvoyant} algorithm that is
$O(1)$-competitive.

Still for a single processor, dropping the assumptions of unbounded
speed and polynomially-bounded power functions, \cite{BCP} gave a
$3$-competitive algorithm for the objective of unweighted flow plus
energy, and a $2$-competitive algorithm for fractional weighted flow
plus energy, both in the uniprocessor case for a large class of power
functions.  The former analysis was subsequently improved
by~\cite{Lachlan2009} to show $2$-competitiveness, along with a matching
lower bound.

Now for multiple processors:
\cite{Lam08} considered the setting of multiple \emph{homogeneous}
processors, where the allowable speeds range between zero and some upper
bound, and the power function is polynomial in this range.  They gave an
algorithm that uses a variant of round-robin for the assignment
policy, and job selection and speed scaling policies from
\cite{BPS}, and showed that this algorithm is scalable for the objective 
of (unweighted) flow plus energy.
Subsequently, \cite{GNS09} showed that a randomized machine selection
algorithm is scalable for weighted flow plus energy (and even more
general objective functions) in the setting of polynomial power 
functions. Both these algorithms provide non-migratory schedules and 
compare their costs with optimal solutions which could even be migratory. 
In comparison, as mentioned above, for the case of polynomial power 
functions, our techniques can give a deterministic constant-competitive 
online algorithm for non-migratory weighted flow time plus energy.
% this extends the former
%work by generalizing to the weighted case, and the latter in being 
%deterministic, though only comparing to a non-migratory optimum. 
(Details appear in the final version.)

In non-power-aware settings, the paper most relevant to this work is
that of~\cite{Chadha2009}, which gives a scalable online algorithm for
minimizing weighted flow on unrelated processors. Their setting is even
more demanding, since they allow the processing requirement of the job
to be processor dependent (which captures a type of heterogeneity that
is orthogonal to the performance energy-efficiency heterogeneity that we
consider in this paper). Our algorithm is based on the same general
intuition as theirs: they assign each new job to the processor that
would result in the least increment in future weighted flow (assuming
HDF is used for job selection), and show that this online algorithm is
scalable using an amortized local competitiveness argument.  However, it
is unclear how to directly extend their potential function to our
power-aware setting; we  had success only in the case that each processor
had allowable speed-power combinations lying in $\{(0,0), (s_i, P_i)\}$.

% \kirknote{I made the above claim less strong.}

%As mentioned above, resource augmentation is required to obtained a
%bounded competitive ratio even for fixed-speed zero-power identical
%processors and \emph{unweighted} flow; SRPT is $O(\min\{\log P, \log
%\frac{n}{m}\})$-competitive, and that there is a matching general lower
%bound on the competitive ratio~\cite{LR}. For weighted flow, resource
%augmentation is required even on a single processor~\cite{BC09}.

%\vspace{-0.15in}
\subsection{Preliminaries}
\label{sec:preliminaries}

\subsubsection{Scheduling Basics.}
We consider only non-migratory schedules, which means that no job can
ever run on one processor, and later run on some other processor. In
general, migration is undesirable as the overhead can be significant. We
assume that preemption is allowed, that is, that jobs may be suspended,
and restarted later from the point of suspension.  It is clear that if
preemption is not allowed, bounded competitiveness is not
obtainable.  The speed is the rate at which work is completed; a job $j$
with size $p_j$ run at a constant speed $s$ completes in $\frac{p_j}{s}$
seconds. A job is completed when all of its work has been  processed.
The flow of a job is the completion time of the job minus the release time of the job.
The weighted flow of a job is the weight of the job times the flow of the job.
For a $t \ge r_j$, let $p_j(t)$ be the remaining unprocessed
work on job $j$ at time $t$. The  fractional weight of job $j$ at this
time is $w_j \frac{p_j(t)}{p_j}$.
The fractional weighted flow of a job is the integral over times
between the job's release time
and its completion time of its fractional weight at that time.
The density of a job is its weight divided by its size.
The job selection policy Highest Density First (HDF) always runs the job of highest density.
The inverse density of a job is its size divided by its weight.

\subsubsection{Power Functions.}
The power function for processor $i$ is denoted by $P_i(s)$, and
specifies the power used when processor is run at speed $s$. We
essentially allow any reasonable power function.  However, we do require
the following minimal conditions on each power function, which we adopt
from~\cite{BCP}.  We assume that the allowable speeds are a countable
collection of disjoint subintervals of $[0, \infty)$. We assume that all
the intervals, except possibly the rightmost interval, are closed on
both ends. The rightmost interval may be open on the right if the power
$P_i(s)$ approaches infinity as the speed $s$ approaches the rightmost
endpoint of that interval.  We assume that $P_i$ is non-negative, and
$P_i$ is continuous and differentiable on all but countably many points.
We assume that either there is a maximum allowable speed $T$, or that
the limit inferior of $P_i(s)/s$ as $s$ approaches infinity is not zero
(if this condition doesn't hold then, then the optimal speed scaling
policy is to run at infinite speed). Using transformations specified
in~\cite{BCP}, we may assume without loss of generality that the power
functions satisfy the following properties:
$P$ is continuous and differentiable,
$P(0)=0$, $P$ is strictly increasing,
$P$ is strictly convex, and  $P$ is unbounded.
We use $Q_i$ to denote $P_i^{-1}$; i.e., $Q_i(y)$ gives us the speed
that we can run processor $i$ at, if we specify a limit of $y$.

\subsubsection{Local Competitiveness and Potential Functions.}
Finally, let us quickly review amortized local competitiveness analysis
on a single processor.
Consider an objective $G$. Let $G_A(t)$ be the increase in the objective
in the schedule for algorithm A at time $t$.  So when $G$ is fractional weighted flow
plus energy, $G_A(t)$ is $P_A^t + w_A^t$, where $P_A^t$ is the power
for A at time $t$ and $w_A^t$ is the fractional weight of the unfinished jobs for A at
time $t$.
Let
OPT be the offline adversary that optimizes $G$.
A is locally $c$-competitive if for all times $t$, if
$G_{A}(t) \le c \cdot G_{OPT}(t)$.
To prove A is $(c+d)$-competitive using an amortized local competitiveness argument,
it suffices to give a potential function
$\Phi(t)$ such that the following conditions hold (see for example~\cite{Pruhs07}).

\begin{OneLiners}
\item[{\bf Boundary condition:}] $\Phi$ is zero before any job is
  released and $\Phi$ is non-negative after all jobs are finished.
\item[{\bf Completion condition:}]
$\Phi$ does not increase due to completions by either A or OPT.
\item[{\bf Arrival condition:}]
$\Phi$ does not increase more than $d \cdot OPT$ due to job arrivals.

\item[{\bf Running condition:}]
  At any time $t$ when no job arrives or is completed,
  \begin{equation}\label{eqn:eq1}
    G_{A}(t) + \frac{d\Phi(t)}{dt} \le c \cdot
    G_{OPT}(t)
  \end{equation}
\end{OneLiners}
The sufficiency of these conditions for proving $(c+d)$-competitiveness
follows from integrating them over time.

\section{Weighted Flow}
\label{sec:weighted-flow-time}

Our goal in this section is to prove Theorem \ref{thm:main1}.
We first show that the online algorithm is $(1+\epsilon)$-speed
$O(\frac{1}{\epsilon})$-competitive for the objective of
\emph{fractional} weighted flow plus energy. Theorem \ref{thm:main1}
then follows since HDF is
$(1+\epsilon)$-speed
$O(\frac{1}{\epsilon})$-competitive for fixed processor speeds~\cite{BLMP1}
for the objective of (integer) weighted flow.

Let \opt be some
optimal schedule minimizing fractional weighted flow.
Let $w_{a,i}^{t}(q)$
denote the total fractional weight of jobs in processor
$i$'s queue that have an inverse density of at least $q$.
Let $w_{a,i}^{t} := w_{a,i}^{t}(0)$ be
the total fractional weight of unfinished jobs in the queue.
Let $w_{a}^{t} := \sum_i w_{a,i}^{t}$ be
the total fractional weight of unfinished jobs in all queues.
Let
$w_{o,i}^{t}(q)$,
$w_{o,i}^{t}$, and
$w_{o}^{t}$ be similarly defined for \opt.
When the time instant being considered is clear, we drop the superscript of $t$ from all variables.

We assume that once \opt
has assigned a job to some processor, it runs the \bcp
algorithm~\cite{BCP} for job selection and speed scaling---i.e., it sets
the speed of the $i^{th}$ processor to $Q_i(w_{o,i})$, and hence
the $i^{th}$ processor uses power $W_{o,i}$,
and uses HDF for job selection.
We can make such an assumption because the results of~\cite{BCP} show that the
fractional weighted flow plus energy of the schedule output by this algorithm is within a factor of two of
optimal.  Therefore, the only real difference between \opt and the online algorithm is the
assignment policy.

%   As mentioned in the introduction, we give a new potential function and a proof that shows that the \bcp algorithm for weighted flow is scalable on a single processor with respect to this potential function; morevoer, we
%show that while our algorithm (which employs a simple greedy-based assignment policy) and \opt may schedule jobs on different
%processors with very different power functions, one can bound the
%increase thus caused to this potential function.

% Likewise, once our assignment policy assigns a job to some processor,
% running BCP on every processor would be locally $O(1)$-competitive
% against the optimal solution to schedule the jobs that have been
% assigned to any processor.  Therefore, the crux of our algorithm would
% be in designing a good assignment policy, and arguing that it is
% $O(1)$-competitive even though our algorithm and \opt may schedule a new
% job on different processors with completely different power functions.

\subsection{The Assignment Policy}

To better understand the online algorithm's assignment policy,
define the ``shadow potential'' for processor $i$ at time $t$ to be
\begin{equation}
  \label{eq:shadow-wtd} \tsty {\widehat\Phi}_{a,i}(t) = \int_{q =
    0}^{\infty} \int_{x = 0}^{w^{t}_{a,i}(q)} \frac{x}{Q_i(x)} \, dx\, dq
\end{equation}
The shadow potential captures (up to a constant factor) the total fractional weighted flow to serve the current set of jobs if no jobs arrive in the
future.
% (see Appendix~\ref{sec:cost-estimate} for an explanation).
Based on this, the online algorithm's assignment policy can alternatively be described as follows:

\medskip \noindent {\bf Assignment Policy.} When a new job with size $p_j$ and weight $w_j$
arrives at time $t$, the assignment policy assigns it
to a processor which would cause the smallest increase in the
shadow potential; i.e. a processor minimizing
\begin{eqnarray*}
&&\int_{q=0}^{d_j} \int_{x=0}^{w^{t}_{a,i}(q) + w_j } \frac{x}{Q_i(x)} \,dx\,dq-
\int_{q=0}^{d_j} \int_{x=0}^{w^{t}_{a,i}(q) } \frac{x}{Q_i(x)} \,dx\,dq\\
&=& \int_{q=0}^{d_j} \int_{x = w^{t}_{a,i}(q)}^{ w^{t}_{a,i}(q) + w_j } \frac{x}{Q_i(x)} \,dx\,dq
\end{eqnarray*}

%Once the job is assigned to a
%processor, it never migrates. Each processor uses the BCP algorithm for
%job selection and speed scaling (i.e., it uses HDF and sets the power
%equal to the fractional remaining weight $w_{a,i}$).

\subsection{Amortized Local Competitiveness Analysis}

%To analyze this algorithm, we introduce a potential function $\Phi$
%and show the following:
%\begin{enumerate}
%\item [(a)] The total increase in potential due to job
%  arrivals is $O(1/\epsilon)\, c(\opt)$, and
%\item [(b)] At any time instant when no job arrives, $\dadt + \dphdt
%  \leq O(1) \dodt$.
%\end{enumerate}
%Here, $\dadt$ at any time $t$ represents the total fractional weight of unfinished
%jobs plus the total power consumed by the online algorithm (i.e the incremental cost
%towards the objective function), and $\dodt$ represents the corresponding quantity
%for the optimal solution \opt.
% Furthermore, if we ensure that $\Phi(0) = \Phi(\infty) = 0$, then we can
% see (by integrating inequality~(\ref{eq:potn-intro})) from $0$ to
% $\infty$ (and accounting for the discrete jumps in potential due to job
% arrivals) that
% $$c(\alg) = \int_{0}^{\infty} \dadt \leq O(1) \int_{0}^{\infty} \dodt
% + O(1/\epsilon) c(\opt) \leq O(1/\epsilon) c(\opt)$$
%As indicated in Section~\ref{sec:preliminaries}, these two properties will imply that
%$c(\alg) = O(1/\epsilon) c(\opt)$, where $c(\cdot)$ denotes the total weighted fractional flow plus energy of a schedule.

We apply a local competitiveness argument as described in subsection \ref{sec:preliminaries}.
Because the online algorithm is using the BCP algorithm on each processor, the power
for the online algorithm is
$\sum_i P_i(Q_i(w_{a,i}))=w_{a}$. Thus $G_A = 2w_{a}$. Similarly, since \opt is using BCP on
each processor $G_{\opt}= 2 w_{o}$.

\subsubsection{Defining the potential function}
For processor $i$, define the potential
\begin{gather}
  \tsty \Phi_i(t) = \const \int_{q=0}^{\infty} \int_{x=0}^{(w_{a,i}^{t}(q) -
    w_{o,i}^{t}(q))_+} \frac{x}{Q_i(x)} \,dx\,dq \label{eq:pot-wf}
\end{gather}
Here $(\cdot)_+ = \max(\cdot, 0)$.
The global potential is then defined to be $\Phi(t) = \sum_i \Phi_i(t)$.
Firstly, we observe that the function $x/Q_i(x)$ is
increasing and subadditive.
% (see Appendix~\ref{sec:subadditive} for a proof). 
Then, the following lemma
% (proof in Appendix~\ref{sec:missing-proofs}) 
will be useful subsequently, the proof of which will appear in the full version of the paper.

\begin{lemma}
  \label{lem:concave-arrival}
  Let $g$ be any increasing subadditive function with $g(0) \geq 0$, and $w_a, w_o, w_j \in \R_{\geq 0}$.
  Then,
  \begin{equation*}
    \tsty \int_{x=w_a}^{w_a + w_j} g(x) \,dx - \int_{x = (w_a - w_o -
      w_j)_{+}}^{(w_a - w_o)_{+}} g(x) \,dx \leq 2 \int_{x=0}^{w_j} g(w_o
    + x) \,dx
  \end{equation*}
\end{lemma}

That the boundary and completion conditions are satisfied are obvious.
In Lemma \ref{lem:arrival-wtd} we prove that the arrival condition holds,
and in Lemma \ref{lem:running-wtd} we prove that the running condition holds.

% \kirknote{I made several algebraic changes to the proof of the following lemma. It should
% absolutely be check again by someone else.}

\begin{lemma}
  \label{lem:arrival-wtd}
  The arrival condition holds with $d = \frac{4}{\epsilon}$.
\end{lemma}
\begin{proof}
Consider a new job $j$ with processing time $p_j$, weight $w_j$ and
inverse density $d_j = p_j/w_j$, which the algorithm assigns to
processor~$1$ while the optimal solution assigns it to processor $2$.
Observe that $\int_{q=0}^{d_j}
\int_{x=w_{o,2}(q)}^{w_{o,2}(q) + w_j} \frac{ x}{Q_2( x)} \,dx \,dq $ is
the increase in \opt's fractional weighted flow due to this new job $j$.
Thus our goal is to prove that the increase in the potential due to
job $j$'s arrival is at most this amount.
The change in the potential $\Delta \Phi$ is:
\begin{equation}
  \label{eq:phi1-wtd}
\nonumber \tsty \const \int_{q=0}^{d_j} \left( \int_{x= (w_{a,1}(q) -
    w_{o,1}(q))_{+}}^{(w_{a,1}(q) - w_{o,1}(q) + w_j)_+}
  \frac{x}{Q_1(x)} \,dx - \int_{x= (w_{a,2}(q) -
    w_{o,2}(q) -w_j)_+}^{(w_{a,2}(q) - w_{o,2}(q))_+} \frac{x}{Q_2(x)}
  \,dx  \right) \,dq
\end{equation}
Now, since $x/Q_1(x)$ is an increasing function we have that
\[
\tsty
\int_{x= (w_{a,1}(q) - w_{o,1}(q))_{+}}^{(w_{a,1}(q) - w_{o,1}(q) +
  w_j)_+} \frac{x}{Q_1(x)} \,dx \leq \int_{x= w_{a,1}(q)}^{w_{a,1}(q) +
  w_j} \frac{x}{Q_1(x)} \,dx
\]
and hence the change of potential can be bounded by
\[
\tsty \const \int_{q=0}^{d_j} \left( \int_{x= w_{a,1}(q)}^{w_{a,1}(q) +
  w_j} \frac{x}{Q_1(x)} \,dx - \int_{x= (w_{a,2}(q) -
  w_{o,2}(q)-w_j)_+}^{(w_{a,2}(q) - w_{o,2}(q))_+} \frac{x}{Q_2(x)}
\,dx \right) \,dq
\]
Since we assigned the job to processor~$1$, we know that
\[ \int_{q=0}^{d_j} \int_{x= w_{a,1}(q)}^{w_{a,1}(q) + w_j}
\frac{x}{Q_1(x)} \,dx\,dq \leq \int_{q=0}^{d_j} \int_{x=
  w_{a,2}(q)}^{w_{a,2}(q) + w_j} \frac{x}{Q_2(x)} \,dx\,dq
\]
Therefore, the change in potential is at most
\begin{align*}
  \Delta \Phi &\leq \tsty \const \int_{q=0}^{d_j} \left( \int_{x=
    w_{a,2}(q)}^{w_{a,2}(q) + w_j} \frac{x}{Q_2(x)} \,dx -
  \int_{x= (w_{a,2}(q) - w_{o,2}(q) - w_j)_+}^{(w_{a,2}(q) - w_{o,2}(q))_+} \frac{x}{Q_2(x)} \,dx  \right) \,dq
\end{align*}
Applying Lemma~\ref{lem:concave-arrival}, we get:
\begin{align*}
  \Delta \Phi & \leq \big( 2 \cdot \const \big) \int_{q=0}^{d_j}
  \int_{x=w_{o,2}(q)}^{w_{o,2}(q) + w_j} \frac{ x}{Q_2( x)} \,dx \,dq
\end{align*}
\end{proof}

\begin{lemma}
  \label{lem:running-wtd}
The running condition holds with constant $c=1+\frac{1}{\epsilon}$.
\end{lemma}

\begin{proof}
Let us consider an infinitesimally small interval $[t, t+dt)$ during
which no jobs arrive and analyze the change in the potential $\Phi(t)$.
Since $\Phi(t) = \sum_i \Phi_i(t)$, we can do this on a per-processor
basis. Fix a single processor $i$, and time $t$.
Let $w_i(q) := (w_{a,i}(q) - w_{o,i}(q))_+$, and
$w_i := (w_{a,i} - w_{o,i})_+$.
Let $q_a$ and $q_o$ denote the inverse densities of the jobs being
executed on processor $i$ by the algorithm and optimal solution
respectively (which are the densest jobs in their respective queues,
since both run HDF). Define $s_a = Q_i(w_{a,i})$ and $s_o =
Q_i(w_{o,i})$. Since we assumed that \opt uses the \bcp algorithm on each
processor,  \opt runs processor $i$ at speed $s_o$.
Since the online algorithm is also using \bcp, but has
$(1+\epsilon)$-speed augmentation, the online algorithms
runs the processor at speed
$(1+\epsilon) s_a$. Hence the
fractional weight of the job the online algorithm works on decreases at
a rate of $s_a (1+\epsilon)/ q_a$.  Therefore, the quantity
$w_{a,i}(q)$ drops by $s_a \,dt (1+\epsilon)/ q_a$ for $q \in [0, q_a]$.
Likewise, $w_{o,i}(q)$ drops by $s_o \,dt/q_o$ for $q \in [0, q_o]$ due to
the optimal algorithm working on its densest job. We consider
several different cases based on the values of $q_o, q_a, w_{o,i}$, and $w_{a,i}$ and establish bounds on $\dphidtside$;
Recall the definition of $\Phi_i(t)$
from equation~(\ref{eq:pot-wf}):
\begin{equation*}
  \Phi_i(t) = \const \int_{q=0}^{\infty} \int_{x=0}^{(w_{a,i}^{t}(q) -
    w_{o,i}^{t}(q))_+} \frac{x}{Q_i(x)} \,dx\,dq
\end{equation*}

\medskip \noindent {\bf Case (1): $w_{a,i} < w_{o,i}$}:
The only
possible increase in potential function occurs due to the decrease in
$w_{o,i}(q)$, which happens for values of $q \in [0, q_o]$.
But for $q$'s in this range, $w_{a,i}(q) \le w_{a,i}$ and $w_{o,i}(q) = w_{o,i}$.
Thus the inner integral is empty, resulting in no increase in potential.
The running condition then holds since $w_{a,i} < w_{o,i}$.

\medskip \noindent {\bf Case (2): $w_{a,i} > w_{o,i}$}:
To quantify the change in potential due to the
online algorithm working, observe that for any $q \in [0, q_a]$, the
inner integral of $\Phi_i$ decreases by
\[
\int_{x= 0}^{w_i(q)} \frac{x}{Q_i(x)} \,dx - \int_{x= 0}^{w_i(q) -
  (1+\epsilon) \frac{s_a \,dt}{q_a}} \frac{x}{Q_i(x)} \,dx
= \frac{w_i(q)}{Q_i(w_i(q))} (1+\epsilon) \frac{s_a \,dt}{q_a}
\]
Here, we have used the fact that $dt$ is infinitisemally small to get the above equality. Hence, the total drop in $\Phi_i$ due to the online algorithm's
processing is
\begin{eqnarray*}
\const \int_{q = 0}^{q_a} \frac{w_i(q)}{Q_i(w_i(q))} (1+\epsilon) \frac{s_a
  \,dt}{q_a} \,dq  &\geq &
  \const \int_{q = 0}^{q_a} \frac{w_i}{Q_i(w_i)} (1+\epsilon) \frac{s_a
  \,dt}{q_a} \,dq\\
&=& \const \frac{w_i}{Q_i(w_i)} (1+\epsilon) s_a \,dt
\end{eqnarray*}
Here, the first inequality holds because $x/Q_i(x)$ is a non-decreasing function, and for all $q \in [0, q_a]$, we have $w_{a,i}(q) = w_{a,i}$ and $w_{o,i}(q) \leq w_{o,i}$ and hence $w_i(q) \geq w_i$.

Now to quantify the increase in the potential due to the optimal algorithm
working: observe that for $q \in [0, q_o]$, the inner integral of
$\Phi_i$ increases by at most
\[ \int_{x = w_{i}(q)}^{w_i(q) + \frac{s_o \,dt}{q_o} } \frac{x}{Q_i(x)}
\,dx = \frac{w_i(q)}{Q_i(w_i(q))} \frac{s_o \,dt}{q_o}
\]
Again notice that we have used that fact that here $dt$ is an infinitesimal period of time that in the limit is zero. Hence the total increase in $\Phi_i$ due to the optimal algorithm's processing
is at most
\[  \const \int_{q = 0}^{q_o} \frac{w_i(q)}{Q_i(w_i(q))} \frac{s_o \,dt}{q_o}
\,dq \leq
\const \int_{q = 0}^{q_o} \frac{w_i}{Q_i(w_i)} \frac{s_o \,dt}{q_o}
\,dq =  \const \frac{w_i}{Q_i(w_i)} s_o \,dt.
\]
Again here, the first inequality holds because $x/Q_i(x)$ is a non-decreasing function, and for all $q \in [0, q_o]$, we have $w_{a,i}(q) \leq w_{a,i}$ and $w_{o,i}(q) = w_{o,i}$ and hence $w_i(q) \leq w_i$.

Putting the two together, the overall increase in $\Phi_i(t)$ can be
bounded by
\begin{align*}
  \dphidt &\leq \const \frac{w_{a,i} - w_{o,i}}{Q_i(w_{a,i} - w_{o,i})}
  \left[ -(1+\epsilon) s_a + s_o \right] \\
  &= \const ( w_{a,i} - w_{o,i} ) \frac {[ -(1+\epsilon) Q_i(w_{a,i}) +
    Q_i(w_{o,i}) ]} {Q_i(w_{a,i} - w_{o,i})}  \\
  &\leq - \const \epsilon (w_{a,i} - w_{o,i}) = - 2 (w_{a,i} -
  w_{o,i}) \label{eq:phi}
\end{align*}
It is now easy to verify that by plugging this bound on $\dphidt$ into the
running condition that one gets a valid inequality.

\medskip \noindent {\bf Case (3): $w_{a,i} = w_{o,i}$}: In this case,
let us just consider the increase due to \opt working. The inner
integral in the potential function starts off from zero (since $w_{a,i}
- w_{o,i} = 0$) and potentially (in the worst case) could increase to
\[ \int_{0}^{ \frac{s_o
    \,dt}{q_o}} \frac{x}{Q_i(x)} \,dx
\] (since $w_{o,i}$ drops by $s_o \,dt/q_o$ and $w_{a,i}$ cannot
increase). However, since $x/Q_i(x)$ is a monotone non-decreasing
function, this is at most
\[  \int_{0}^{ \frac{s_o \,dt}{q_o}} \frac{w_{o,i}}{Q_i(w_{o,i})} \,dx =
\frac{s_o \,dt}{q_o} \frac{w_{o,i}}{Q_i(w_{o,i})}
\]
Therefore, the total increase in the potential $\Phi_i(t)$ can be
bounded by
\[ \const \int_{q = 0}^{q_o} \frac{w_{o,i}}{Q_i(w_{o,i})}
\frac{s_o \,dt}{q_o} \,dq = \const s_o \,dt \frac{w_{o,i}}{Q_i(w_{o,i})} =
\const w_{o,i} \,dt
\]
It is now easy to verify that by plugging this bound on $\dphidt$
into the
running condition,
and using the fact that $w_{a,i} = w_{o,i}$,
one gets a valid inequality.
\end{proof}

\section{Algorithm for Unweighted Flow}
\label{dsec:unweighted-flow}

In this section, we give an immediate assignment based scheduling policy
and show that it is $O(1/\epsilon)$-competitive against a non-migratory
adversary for the objective of \emph{unweighted} flow plus energy,
assuming the online algorithm has resource augmentation of
$(1+\epsilon)$ in speed. Note that this result has a better
competitiveness than the result for weighted flow from
Section~\ref{sec:weighted-flow-time}, but holds only for the unweighted case.

We begin by giving intuition behind our algorithm, which is again
similar to that for the weighted case. Let \opt be some optimal
schedule. 
However, for the rest of the section, we assume that on a single machine, the optimal scheduling algorithm for minimizing sum of flow times plus energy on a single machine
is that of Andrew et al.\cite{Lachlan2009} which sets the power at any time to be $Q(n)$ when there are $n$ unfinished jobs, and processes jobs according to SRPT.
Since we know that this \alw algorithm~\cite{Lachlan2009} is
$2$-competitive against the optimal schedule on a single processor, we
will imagine that, once \opt has assigned a job to some processor, it
uses the \alw algorithm on each processor. Likewise, once our assignment
policy assigns a job to some processor, our algorithm also runs the \alw algorithm on each processor.
Therefore, just like the weighted case, the crux of our algorithm is in designing a good assignment policy,
and arguing that it is $O(1)$-competitive even though our algorithm and
\opt may schedule a new job on different processors with completely
different power functions.

\subsection{Algorithm}
\label{dsec:algorithm-uf}

Our algorithm works as follows: Each processor maintains a queue of jobs
that have currently been assigned to it.  At some time instant $t$, for
any processor $i$, let $n_{a,i}^{t}(q)$ denote the number of jobs in processor
$i$'s queue that have a remaining processing time of at least $q$. Let
$n_{a,i}^{t}$ denote the total number of unfinished jobs in the queue. Also,
let us define the \emph{shadow potential} for processor $i$ at this time
$t$ as
\begin{equation}
  \label{deq:shadow} \tsty {\widehat\Phi}_{a,i}(t) = \int_{q = 0}^{\infty} \sum_{j =
    1}^{n^{t}_{a,i}(q)} \frac{j}{Q_i(j)} \,dq
\end{equation}
Note that the shadow potential
${\widehat\Phi}_{a,i}(t)$ is the total future cost of the online
algorithm (up to a constant factor) assuming no jobs arrive after this
time instant, and the online algorithm runs the \alw algorithm on all
processors (i.e., the job selection is SRPT, and the processor is run at
a speed of $Q_i(n_{a,i}^{t})$). Now our algorithm is the following:
% Then, it can be seen that the total future
% cost incurred by the online algorithm is (within a constant factor)
% identical to the value ${\Phi}_{a,i}(t)$ at this time $t$.
\begin{shadebox}
  \noindent When a new job arrives, the assignment policy assigns it to a
  processor which would cause the smallest increase in the ``shadow
  potential''; i.e., a processor minimizing
  \[
  \int_{q=0}^{p} \sum_{j=1}^{n^{t}_{a,i}(q) + 1 } \frac{j}{Q_i(j)} \,dq -
  \int_{q=0}^{p} \sum_{j=1}^{n^{t}_{a,i}(q) } \frac{j}{Q_i(j)} \,dq =
  \int_{q=0}^{p} \frac{(n^{t}_{a,i}(q) + 1)}{Q_i(n^{t}_{a,i}(q) + 1)} \,dq
  \]
  The job selection on each processor is SRPT (Shortest Remaining
  Processing Time), and we set the power of processor $i$ at time $t$ to $n_{a,i}^{t}$.  Once the job is
  assigned to a processor, it is never migrated.
\end{shadebox}

% \subsubsection{Running Condition}
% Once a job is sent to a processor, the assignment is fixed, i.e. the
% solution is non-migratory.  On each processor $i$, the algorithm schedules
% jobs according to the SRPT policy and determines the speed to be
% $Q_i(n_{a,i} + 2)$ at any time instant when $n_{a,i}$ unfinished jobs
% remain on the processor $i$.

\subsection{The Amortized Local-Competitive Analysis}
\label{dsec:analysis-uf}

We again employ a potential function based analysis, similar to the one in Section~\ref{sec:weighted-flow-time}.
% Furthermore, if we ensure that $\Phi(0) = \Phi(\infty) = 0$, then we can
% see (by integrating inequality~(\ref{eq:potn-intro})) from $0$ to
% $\infty$ (and accounting for the discrete jumps in potential due to job
% arrivals) that
% $$c(\alg) = \int_{0}^{\infty} \dadt \leq O(1) \int_{0}^{\infty} \dodt
% + O(1/\epsilon) c(\opt) \leq O(1/\epsilon) c(\opt)$$
% As indicated in Section~\ref{sec:preliminaries}, this will imply that
% $c(\alg) = O(1/\epsilon) c(\opt)$.

\subsubsection{The Potential Function.}

We now describe our potential function $\Phi$.  For time $t$ and
processor $i$, recall the definitions $n_{a,i}^{t}$ and $n_{a,i}^t(q)$ given above;
analogously define $n_{o,i}^{t}$ as the number of unfinished jobs assigned
to processor $i$ by the optimal solution at time $t$, and $n_{o,i}^{t}(q)$ to be the
number of these jobs with remaining processing time at least $q$.
Henceforth, we will drop the superscript $t$ from these terms whenever the time instant $t$ is clear from the context.

Now, we define the
global potential function to be $\Phi(t) = \sum_i \Phi_i(t)$, where
$\Phi_i(t)$ is the potential for processor $i$ defined as:
\begin{gather}
  \tsty \Phi_i(t) = \constun \int_{q=0}^{\infty} \sum_{j=1}^{(n_{a,i}^{t}(q) -
    n_{o,i}^{t}(q))_+} j/Q_i(j) \,dq \label{deq:pot-uf}
\end{gather}
 Recall that $(x)_+ = \max(x,0)$, and $Q_i = P_i^{-1}$.
%  $( n_{a,i}(q) - n_{o,i}(q) )_+$ captures the difference in
% the number of active jobs of remaining size at least $q$ which are
% assigned to processor $i$, between the algorithm and the optimal
% solution.
Notice that if the optimal solution has no jobs remaining on processor $i$ at time
$t$, we get $\Phi_i(t)$ is (within a constant off) simply ${\widehat\Phi}_{a,i}(t)$.
% We then define the global potential function to
% be
% \[ \tsty  \]

\subsubsection{Proving the Arrival Condition.}

We now show that the increase in the potential $\Phi$ is bounded (up to
a constant factor) by the increase in the future optimal cost when a new
job arrives. Suppose a new job of size $p$ arrives at time $t$, and suppose
 the online algorithm assigns it
to processor~$1$ while the optimal solution assigns it to processor
$2$. Then $\Phi_1$ increases since $n_{a,1}(q)$ goes up by $1$ for all
$q \in [0,p]$, $\Phi_2$ could decrease due to $n_{o,2}(q)$ dropping by
$1$ for all $q \in [0,p]$, and $\Phi_i$ (for $i \notin \{1,2\}$) does
not change.

Let us first assume that $n_{a,i}(q) \geq n_{o,i}(q)$ for all $q \in
[0,p]$ and for $i \in \{1,2\}$; we will show below how to remove this
assumption. Under this assumption, the total change in potential $\Phi$
is
\[
\tsty \constun \int_{q=0}^{p} \big( \frac{n_{a,1}(q) - n_{o,1}(q)+1}{Q_1(
  n_{a,1}(q) - n_{o,1}(q) + 1)} - \frac{n_{a,2}(q) - n_{o,2}(q) }{Q_2(
  n_{a,2}(q) - n_{o,2}(q))} \big) \,dq
\]
But since $x/Q(x)$ is increasing this is less than
\begin{equation}
  \label{deq:phi1}
\tsty  \constun \int_{q=0}^{p} \big( \frac{n_{a,1}(q) + 1}{Q_1( n_{a,1}(q) + 1)}  -
  \frac{n_{a,2}(q) - n_{o,2}(q) }{Q_2( n_{a,2}(q) - n_{o,2}(q))} \big) \,dq
\end{equation}
By the greedy choice of processor $1$ (instead of $2$), this is less
than
\begin{equation}
  \label{deq:phi2}
\tsty  \constun \int_{q=0}^{p} \big( \frac{n_{a,2}(q) + 1}{Q_2( n_{a,2}(q) + 1)}  -
  \frac{n_{a,2}(q) - n_{o,2}(q) }{Q_2( n_{a,2}(q) - n_{o,2}(q) )} \big) \,dq
\end{equation}
Now, since $x/Q(x)$ is subadditive this is less than $\constun
\int_{q=0}^{p} \frac{n_{o,2}(q) + 1}{Q_2( n_{o,2}(q) + 1)} \,dq$, which in
turn is (within a factor of $\constun$) precisely the \emph{increase in
  the future cost} incurred by the optimal solution, since we had
assumed that \opt also runs the \alw algorithm on its processors.

Now suppose $n_{a,1}(q) < n_{o,1}(q)$ for some $q \in [0,p]$. There
is no increase in the inner sum of $\Phi_1$ for such values of $q$, and
hence we can trivially upper bound this zero increase by $\constun
\int_{q=0}^{p} \frac{n_{a,1}(q) + 1}{Q_1( n_{a,1}(q) + 1)} \leq \constun
\int_{q=0}^{p} \frac{n_{a,2}(q) + 1}{Q_2( n_{a,2}(q) + 1)}$. And to
discharge the assumption for processor $2$, note that if $n_{a,2}(q) <
n_{o,2}(q)$ for some $q$, there is no decrease in the inner sum of
$\Phi_2$ for this value of $q$, but in this case we can simply use
$\frac{n_{a,2}(q) + 1}{Q_1( n_{a,2}(q) + 1)} \leq \frac{n_{o,2}(q) +
  1}{Q_1( n_{o,2}(q) + 1)}$ for such values of $q$.
% $$ \big( \frac{n_{a,2}(q) + 1}{Q_1( n_{a,2}(q) + 2)}  -
% \frac{n_{a,2}(q) - n_{o,2}(q) }{Q_2( n_{a,2}(q) - n_{o,2}(q) +1 )} \big) \leq \frac{n_{o,2}(q) + 1}{Q_2( n_{o,2}(q) + 2)}
% $$
% when going from equation~(\ref{eq:phi2}) to~(\ref{eq:phi3}), we can simply sa
Therefore, we get the same bound of \[ \constun \int_{q=0}^{p}
\frac{n_{o,2}(q) + 1}{Q_2( n_{o,2}(q) + 1)} \,dq \] on the increase in all
cases, thus proving the following lemma.
\begin{lemma}
  \label{dlem:arrival} The arrival condition holds for the unweighted
  case with $d = \constun$.
\end{lemma}
\subsubsection{Proving the Running Condition.}
In this section, our goal is to analyze the change in $\Phi$ in an
infinitesimally small time interval $[t, t+ dt)$ and compare $d\Phi/dt$
to $d\alg/dt$ and $d\opt/dt$. We do this on a per-processor basis; let
us focus on processor $i$ at time $t$. Recall (after dropping the $t$ superscripts) the definitions of
$n_{a,i}(q), n_{o,i}(q), n_{a,i}$ and $n_{o,i}$ from above, and define
$n_i(q) := (n_{a,i}(q) - n_{o,i}(q))_+$. Finally, let $q_a$ and $q_o$
denote the remaining sizes of the jobs being worked on by the algorithm
and optimal solution respectively at time $t$; recall that both of them
use SRPT for job selection.  Define $s_a = Q_i(n_{a,i})$ and $s_o =
Q_i(n_{o,i})$ to be the (unaugmented) speeds of processor $i$
according to the online algorithm and the optimal algorithm
respectively---though, since we assume resource augmentation, our
processor runs at speed $(1+ \epsilon) Q_i(n_{a,i})$.  Hence
$n_{a,i}(q)$ drops by $1$ for $q \in (q_a - (1+\epsilon) s_a \,dt, q_a]$
and $n_{o,i}(q)$ drops by $1$ for $q \in (q_o - s_o \,dt, q_o]$ for the
optimal algorithm. Let $I_a := (q_a - (1+\epsilon) s_a \,dt, q_a]$ and
$I_o := (q_o - s_o \,dt, q_o]$ denote these two intervals. Let us consider
some cases:

% Now, by the definition of the potential function, we have $\Phi_i(t) = \int_{q=0}^{\infty} \sum_{j=0}^{n_i(q)} j/Q_i(j)  \,dq$. We consider several different cases and establish useful bounds on $\dphidt$.

\medskip \noindent {\bf Case (1): $n_{a,i} < n_{o,i}$}: The increase in
potential function may occur due to $n_{o,i}(q)$ dropping by $1$ in $q
\in I_o$. However, since $n_{a,i} < n_{o,i}$, it follows that
$n_{a,i}(q) \leq n_{a,i} < n_{o,i} = n_{o,i}(q)$ for all $q \in I_o$;
the equality follows from \opt using SRPT.  Consequently, even with
$n_{o,i}(q)$ dropping by $1$, $n_{a,i}(q) - n_{o,i}(q) \leq 0$ and there
is no increase in potential, or equivalently $\dphidtside \leq 0$.
Hence, in this case, $4 n_{a,i} + \dphidtside \leq 4 n_{o,i}$.

\medskip \noindent {\bf Case (2a): $n_{a,i} \geq n_{o,i}$, and $q_a <
  q_o$}: For $q \in I_a$, the inner summation of $\Phi_i$ drops by $
\frac{n_{a,i}(q) - n_{o,i}(q)}{Q_i(n_{a,i}(q) - n_{o,i}(q))}$, since
$n_{a,i}(q)$ decreases by~$1$. Moreover, $n_{a,i}(q) = n_{a,i}$ and
$n_{o,i}(q) = n_{o,i}$, because both \alg and \opt run
SRPT, and we're considering $q \leq q_a < q_o$.  For $q \in I_o$, the
inner summation of $\Phi_i$ increases by $\frac{n_{a,i}(q) - (n_{o,i}(q)
  - 1)}{Q_i( n_{a,i}(q) - (n_{o,i}(q) - 1))}$.  However, we have
$n_{a,i}(q) \leq n_{a,i} - 1$ and $n_{o,i}(q) = n_{o,i}$ because $q_a <
q_o$, and $n_{o,i}(q) = n_{o,i}$ because \opt runs SRPT.  Therefore the
increase is at most $\frac{n_{a,i} - n_{o,i}}{Q_i(n_{a,i} - n_{o,i})}$
for all $q \in I_o$. Combining these two, we get
\begin{align*}
  \tsty \dphidt &\leq \tsty \constun \frac{n_{a,i} - n_{o,i}}{Q_i(n_{a,i}
    - n_{o,i})} \big[ -(1+\epsilon) s_a + s_o \big] \\
  &= \tsty \constun \big( n_{a,i} - n_{o,i} \big) \frac {\big[ -(1+\epsilon) Q_i(n_{a,i}) + Q_i(n_{o,i}) \big]} {Q_i(n_{a,i} - n_{o,i} )}  \\
  &\leq \tsty \constun \big( n_{a,i} - n_{o,i} \big)\frac{ -\epsilon\,
    Q_i(n_{a,i}) } {Q_i(n_{a,i} - n_{o,i})} \leq - 4 (n_{a,i} -
  n_{o,i}),
\end{align*}
where we repeatedly use that $Q_i(\cdot)$ is a non-decreasing function.
This implies that $4 n_{a,i} + \dphidtside \leq 4 n_{o,i}$.

\medskip \noindent {\bf Case (2b): $n_{a,i} \geq n_{o,i}$, and $q_a >
  q_o$}: In this case, for $q \in I_o$, the inner summation of $\Phi_i$
increases by $ \frac{n_{a,i}(q) - (n_{o,i}(q) - 1)}{Q_i( n_{a,i}(q) -
  (n_{o,i}(q) - 1))}$. Also, we have $n_{a,i}(q) = n_{a,i}$ and
$n_{o,i}(q) = n_{o,i}$, because $q \leq q_o < q_a$ and both algorithms
run SRPT. Therefore the overall increase in potential function can be
bounded by $\constun \frac{n_{a,i} - n_{o,i} + 1}{Q_i(n_{a,i} - n_{o,i} + 1)} \;
s_o \,dt $. Moreover, for $q \in I_a$, the inner summation of $\Phi_i$
drops by $ \frac{n_{a,i}(q) - n_{o,i}(q)}{Q_i(n_{a,i}(q) - n_{o,i}(q))}$. Also, $n_{a,i}(q) = n_{a,i}$ and $n_{o,i}(q) \leq n_{o,i} - 1$,
because $q_o < q_a$ and there was a job of remaining size $q_o$ among
the optimal solution's active jobs. Thus the potential function drops by
at least $ \constun \frac{n_{a,i} - n_{o,i} + 1}{Q_i(n_{a,i} - n_{o,i} + 1)}\; (1
+\epsilon) s_a \,dt $, since $x/Q_i(x)$ is a non-decreasing function.
Combining these terms,
\begin{align*}
  \tsty \dphidt &\leq \tsty \constun \frac{n_{a,i} - n_{o,i} +
    1}{Q_i(n_{a,i} - n_{o,i} + 1)} \big[ -(1+\epsilon) s_a + s_o \big]
  \\
  & = \constun \big( n_{a,i} - n_{o,i} + 1 \big) \frac {[ -(1+\epsilon)
    Q_i(n_{a,i}) + Q_i(n_{o,i}) ]} {Q_i(n_{a,i} - n_{o,i} + 1)}  \\
  &\tsty \leq - \constun \epsilon (n_{a,i} - n_{o,i} + 1) \leq - 4
  (n_{a,i} - n_{o,i})
\end{align*}
In the above, we have used the fact that $n_{o,i} \geq 1$, and consequently, $Q_i(n_{a,i}) \geq Q_i(n_{a,i}- n_{o,i} + 1)$. Therefore, in this case too we get $4 n_{a,i} + \dphidtside \leq 4 n_{o,i}$.

\medskip \noindent {\bf Case (2c): $n_{a,i} \geq n_{o,i}$, and $q_a =
  q_o$}: Since $n_{a,i} \geq n_{o,i}$, and $Q_i$ is an increasing
function, $s_a \geq s_o$ and thus $I_o \subset I_a$. For $q$ in the
interval $I_o$, the term $n_{o,i}(q)$ drops by $1$ \emph{and} the term
$n_{a,i}(q)$ drops by $1$, and therefore there is no change in
$n_{a,i}(q) - n_{o,i}(q)$. For $q \in I_a \setminus I_o$, the inner
summation for $\Phi_i$ drops by $\frac{n_{a,i}(q) - n_{o,i}(q) }{Q_i(
  n_{a,i}(q) - n_{o,i}(q))}$. Also, $n_{a,i}(q) = n_{a,i}$ and
$n_{o,i}(q) = n_{o,i}$, and the decrease in potential function is $\constun (
(1+\epsilon) s_a \,dt - s_o \,dt) \; \frac{n_{a,i} - n_{o,i} }{Q_i(n_{a,i} -
  n_{o,i})}$. But the analysis in Case~(2a) implies that $4n_{a,i}
+ \dphidtside \leq 4 n_{o,i}$ in this case as well.

Summing over all $i$, we
get
\begin{lemma}
  \label{dlem:running}
  The running condition holds for the unweighted case with constant $4$.
  At any time $t$ when there are no job arrivals, we have
\end{lemma}
Combining Lemmas~\ref{dlem:arrival} and~\ref{dlem:running} with the
standard potential function argument indicated in
Section~\ref{sec:preliminaries}, we get the following theorem.
\begin{theorem}
  \label{dthm:unwtd}
  There is a $(1+\epsilon)$-speed $O(1/\epsilon)$-competitive
  immediate-assignment algorithm to minimize the total flow plus energy on
  heterogeneous processors with arbitrary power functions.
\end{theorem}

\bigskip
\noindent
{\bf Acknowledgments:} We thank Sangyeun Cho and Bruce Childers for helpful discussions
about heterogeneous multicore processors, and also Srivatsan Narayanan for several useful discussions.

\bibliographystyle{plain}
\bibliography{icalp}

\appendix

\section{Estimating the Future Cost of BCP}
\label{sec:cost-estimate}
Suppose we have a set of $n$ jobs, with weight $w_j$ and processing time
$p_j$ for job $j$ ($1 \leq j \leq n$) such that $p_1/w_1 \leq p_2/w_2 \leq \ldots \leq p_n/w_n$. Let $d_j = p_j/w_j$ denote the inverse density of a job $j$.
We now explain how we get the estimate of the future cost of this configuration
when we run the BCP algorithm, i.e. HDF at a speed of $Q(w^t)$, where $w^t$ is the total fractional weight
of unfinished jobs at time $t$.

Firstly, observe that by virtue of our algorithm running HDF, it schedules job $1$ followed by job $2$, etc.
Also, as long as the algorithm is running job $1$, it runs the processor at speed
$Q(W_{\geq 2} + \widetilde{w_1}(t))$, where $W_{\geq 2} := w_2 + w_3 + \ldots w_n$ and $\widetilde{w_1}(t)$ is
the fractional weight of job $1$ remaining.
Secondly, since our algorithm always uses power equal to the fractional weight remaining, the rate of increase of the objective function at any time $t$ is simply $2 w^t$.
Therefore, the following equations immediately follow:
\begin{eqnarray}
\nonumber G_A(t) &=& \frac{dA}{dt} = 2\big(W_{\geq 2} + \widetilde{w_1}(t)\big) \\
\nonumber \frac{ d \widetilde{w_1}(t)}{dt} &=& -\bigg( \frac{w_1}{p_1}\bigg) Q\bigg( W_{\geq 2} + \widetilde{w_1}(t) \bigg) \\
\nonumber \Rightarrow \frac{dA}{ d \widetilde{w_1}(t)} &=& - 2 \bigg( \frac{p_1}{w_1} \bigg) \frac{W_{\geq 2} + \widetilde{w_1}(t)}{Q\big( W_{\geq 2} + \widetilde{w_1}(t) \big)} \\
\nonumber \Rightarrow A_1 &=&  - 2 \int_{x = W_{\geq 2} + w_1}^{W_{\geq 2}} d_1 \frac{W_{\geq 2} + x}{Q\big( W_{\geq 2} + x \big)} \,dx
\end{eqnarray}
That is, the total cost incurred while job $1$ is being scheduled is
$$
 2 \int_{x = W_{\geq 2}}^{W_{\geq 2} + w_1} d_1 \frac{W_{\geq 2} + x}{Q\big( W_{\geq 2} + x \big)} \,dx = 2 \int_{q = 0}^{d_1} \int_{x = W_{\geq 2}}^{W_{\geq 2} + w_1} \frac{W_{\geq 2} + x}{Q\big( W_{\geq 2} + x \big)} \,dx \,dq
 $$
Similarly, while any job $i$ is being scheduled, we can use the same arguments as above to show that the total fractional flow incurred is
 $$
 2 \int_{x = W_{\geq (i+1)}}^{W_{\geq (i+1)} + w_i} d_i \frac{W_{\geq (i+1)} + x}{Q\big( W_{\geq (i+1)} + x \big)} \,dx = 2 \int_{q = 0}^{d_i} \int_{x = W_{\geq (i+1)}}^{W_{\geq (i+1)} + w_i} \frac{W_{\geq (i+1)} + x}{Q\big( W_{\geq (i+1)} + x \big)} \,dx \,dq
 $$
 Summing over $i$, the total fractional flow incurred by our algorithm
 is
 $$
2 \sum_{i =1}^{n} \int_{q = 0}^{d_i} \int_{x = W_{\geq (i+1)}}^{W_{\geq (i+1)} + w_i} \frac{W_{\geq (i+1)} + x}{Q\big( W_{\geq (i+1)} + x \big)} \,dx \,dq
 $$
Rearranging the terms, it is not hard to see (given $d_1 \leq d_2 \leq \ldots \leq d_n$) that this is equal to
$$
2 \int_{q = 0}^{\infty} \int_{x = 0}^{w(q)} \frac{x}{Q(x)} \,dx \,dq
$$
where $w(q)$ is the total weight of jobs with inverse density at least $q$.

\section{Subadditivity of x/Q(x)}
\label{sec:subadditive}
Let $Q(x): \R_{\geq 0} \rightarrow \R_{\geq 0}$ be any concave function such that $Q(0) \geq 0$ and $Q'(x) \geq 0$ for all $x \geq 0$, and let $g(x) = x/Q(x)$. Then the following facts are true about $g(\cdot)$.
\begin{enumerate}
\item[(a)] $g(\cdot)$ is non-decreasing. That is, $g(y) \geq g(x)$ for all $y \geq x$.
\item[(b)] $g(\cdot)$ is subadditive. That is, $g(x) + g(y) \geq g(x+y)$ for all $x, y \in \R_{\geq 0}$
\end{enumerate}

To see why the first is true, consider $x$ and $y = \lambda x$ for some $\lambda \geq 1$. Then,
showing (a) is equivalent to showing
$$
\frac{\lambda x}{Q(\lambda x)} \geq \frac{x}{Q(x)}
$$
But this reduces to showing $Q(\lambda x) \leq \lambda Q(x)$ which is true because $Q(\cdot)$ is a concave function. To prove the second property, we first observe that the function $1/Q(x)$ is convex. This is because the second derivative of $1/Q(x)$ is
$$
\frac{-Q(x)^2 Q''(x) + 2 Q(x) Q'(x)^2}{Q(x)^4}
$$
which is always non-negative for all $x \geq 0$, since $Q(x)$ is non-negative and $Q''(x)$ is non-positive for all $x\geq 0$. Therefore, since $1/Q(\cdot)$ is convex, it holds for any $x$, $y$, and $\alpha \geq 0$, $\beta \geq 0$ that
$$
\frac{\frac{\alpha}{Q(x)} + \frac{\beta}{Q(y)}}{\alpha + \beta} \geq \frac{1}{Q(\frac{\alpha x + \beta y}{\alpha + \beta})}
$$
Plugging in $\alpha = x$ and $\beta = y$, we get
$$
\frac{\frac{x}{Q(x)} + \frac{y}{Q(y)}}{x + y} \geq \frac{1}{Q(\frac{x^2 + y^2}{x + y})}
$$
which implies
$$
\frac{x}{Q(x)} + \frac{y}{Q(y)} \geq \frac{x + y}{Q(\frac{x^2 + y^2}{x + y})}
$$
But since $Q(\cdot)$ is non-decreasing, we have $Q(x + y) \geq Q(\frac{x^2 + y^2}{x + y})$ and hence
$$
\frac{x}{Q(x)} + \frac{y}{Q(y)} \geq \frac{x + y}{Q(x + y)}
$$

\section{Missing Proofs}
\label{sec:missing-proofs}

\begin{proofof}{Lemma~\ref{lem:concave-arrival}}
We first show that
$$\nonumber \int_{x=w_a}^{w_a + w_j} g(x) \,dx - \int_{x = (w_a - w_o - w_j)_{+}}^{(w_a - w_o)_{+}} g(x) \,dx \leq \int_{x=0}^{w_j} g(w_o + w_j) \,dx$$ and then argue that $\int_{x=0}^{w_j} g(w_o + w_j) \,dx \leq 2 \int_{x=0}^{w_j} g(w_o + x) \,dx$ because $g(\cdot)$ is subadditive.
To this end, we consider several cases and prove the lemma.
Suppose $w_a$ is such
that $w_a \geq w_o + w_j$: in this case we can discard the
$(\cdot)_+$ operators on all the limits to get
\begin{align*}
  & \int_{x = w_a}^{w_a + w_j} g(x) \, dx \; - \;
  \int_{x = w_a - w_o  - w_j}^{w_a - w_o} g(x)\,dx
  \\
  &= \int_{x = 0}^{w_j} \bigg( g(w_a + x) - g(w_a - w_o
  - w_j + x) \bigg) \,dx \; \leq \; \int_{x = 0}^{w_j} g(w_o + w_j) \,dx
\end{align*}
Here, the final inequality follows because $g(\cdot)$ is a subadditive function.
On the other hand, suppose it is the case that $w_a \leq w_o$, then both limits $(w_a - w_o)_{+}$ and $(w_a - w_o - w_j)_{+}$ are zero, and therefore we only need to bound $ \int_{x = w_a}^{w_a + w_j} g(x) \, dx$, which can be done as follows:
\begin{gather*}
  \int_{x = w_a}^{w_a + w_j} g(x) \, dx = \int_{x =
    0}^{w_j} g(w_a + x) \, dx \leq \int_{x = 0}^{w_j}
  g(w_o + w_j) \, dx
\end{gather*}
Finally, if $w_a = w_o + \delta$ for some $\delta \in (0, w_j)$,
we first observe that $\int_{x = (w_a - w_o - w_j)_{+}}^{(w_a - w_o)_{+}} g(x) dx$ simplifies to $\int_{x = 0}^{\delta} g(x) dx$. Therefore, we are interested in bounding
\begin{align*}
  & \int_{x = w_a}^{w_a + w_j} g(x) \, dx \; - \; \int_{x =
    0}^{\delta} g(x)\,dx = \int_{x = w_a}^{w_a + w_j - \delta}
  g(x) \, dx + \int_{x = w_a + w_j - \delta}^{w_a + w_j}
  g(x) \, dx \; - \; \int_{x =
    0}^{\delta} g(x)\,dx \\
  &\leq (w_j - \delta) g(w_a + w_j - \delta) + \int_{x =
    0}^{\delta} g(w_a + w_j - \delta) \,dx \; \leq \; \int_{x =
    0}^{w_j} g(w_o + w_j) \,dx
\end{align*}
Here again, in the second to last inequality, we used the fact that $g(\cdot)$ is subadditive and therefore $g(w_a + w_j - \delta + x) - g(x) \leq g(w_a + w_j - \delta)$, for all values of $x \geq 0$; hence we get $\int_{x = w_a + w_j  - \delta}^{w_a + w_j} g(x) \, dx - \int_{x = 0}^{\delta} g(x) \, dx \leq \int_{x = 0}^{\delta} g(w_a + w_j  - \delta) \, dx$.

\medskip \noindent To complete the proof, we need to show that $\int_{x=0}^{w_j} g(w_o + w_j) \,dx \leq 2 \int_{x=0}^{w_j} g(w_o + x) \,dx$. To see this, consider the following sequence of steps:

For any $x \in [0, w_j]$, since $g$ is subadditive, we have
$$
g(w_o + w_j - x) + g(x) \geq g(w_o + w_j)
$$
Integrating both sides from $x = 0$ to $w_j$ we get
$$
\int_{x = 0}^{w_j} g(w_o + w_j - x) \,dx  + \int_{x=0}^{w_j} g(x) \,dx \geq \int_{x = 0}^{w_j} g(w_o + w_j) \,dx
$$
which is, by variable renaming, equivalent to
$$
\int_{y = 0}^{w_j} g(w_o + y) \,dy  + \int_{x=0}^{w_j} g(x) \,dx \geq \int_{x = 0}^{w_j} g(w_o + w_j) \,dx
$$
But since $g(\cdot)$ is non-decreasing, we have $\int_{x=0}^{w_j} g(x) \,dx \leq \int_{x=0}^{w_j} g(w_o + x) \,dx$ and therefore
$$
\int_{y = 0}^{w_j} g(w_o + y) \,dy  + \int_{x=0}^{w_j} g(w_o + x) \,dx \geq \int_{x = 0}^{w_j} g(w_o + w_j) \,dx
$$
which is what we want.
\end{proofof}

\end{document}